\documentclass[preprint,12pt]{elsarticle}
\usepackage[utf8]{inputenc}
\usepackage{lmodern}
\usepackage{microtype}
\usepackage{amsmath,amssymb,amsthm,mathtools}
\usepackage{booktabs}
\usepackage{array}
\usepackage{enumitem}
\usepackage{placeins}
\usepackage{xcolor}
\usepackage{hyperref}

\hypersetup{
  colorlinks=true,
  linkcolor=blue!50!black,
  citecolor=blue!50!black,
  urlcolor=blue!60!black,
  pdftitle={When More Generators Hurt: Shellsort on Full Product Grids},
  pdfauthor={Ziqi Zhao and Qingjian Ni}
}

\usepackage[nameinlink,capitalise]{cleveref}

\newtheorem{theorem}{Theorem}[section]
\newtheorem{proposition}{Proposition}[section]
\newtheorem{lemma}{Lemma}[section]
\newtheorem{corollary}{Corollary}[section]
\theoremstyle{definition}

\theoremstyle{remark}
\newtheorem{remark}{Remark}[section]
\AddToHook{env/theorem/begin}{\crefalias{section}{theorem}}
\AddToHook{env/proposition/begin}{\crefalias{section}{proposition}}
\AddToHook{env/lemma/begin}{\crefalias{section}{lemma}}
\AddToHook{env/corollary/begin}{\crefalias{section}{corollary}}
\AddToHook{env/definition/begin}{\crefalias{section}{definition}}
\AddToHook{env/remark/begin}{\crefalias{section}{remark}}
\makeatletter
\let\c@proposition\c@theorem
\let\c@lemma\c@theorem
\let\c@corollary\c@theorem
\let\c@definition\c@theorem
\let\c@remark\c@theorem
\makeatother

\newcommand{\bbT}{\mathbb T}
\newcommand{\ee}{\mathrm e}
\newcommand{\normfrac}[1]{\left\lVert #1\right\rVert_{\bbT}}

\journal{Information and Computation}

\begin{document}

\begin{frontmatter}

\title{When More Generators Hurt: Shellsort on Full Product Grids}

\author[seu]{Ziqi Zhao}
\ead{ziqizhao@seu.edu.cn}

\author[seu]{Qingjian Ni\corref{cor1}}
\cortext[cor1]{Corresponding author}
\ead{nqj@seu.edu.cn}

\affiliation[seu]{
  organization={School of Computer Science and Engineering, Southeast University},
  city={Nanjing},
  country={China}
}

\begin{abstract}
Shellsort repeatedly runs insertion sort with decreasing gaps, so its
worst-case cost depends on the gap sequence.  Pratt's $2^u3^v$ sequence, one
of the few systematic constructions with a proven $O(n\log^2 n)$ bound,
includes every product below $n$ of two base numbers, or generators.  We ask
whether adding more base numbers, and thus more intermediate gaps, can improve
this full product grid.

We show that it cannot when every product is retained and each base is at most
a fixed power of the smallest.  With $r$ independent bases---different
exponent choices give different products---and $\Theta(\log n)$ gaps, the best
possible worst-case cost is
\[
  n\exp\!\left(\Theta((\log n)^{1-1/r})\right).
\]
Thus two bases give the exponent $\sqrt{\log n}$, whereas three give
$(\log n)^{2/3}$: more bases are worse.  With a budget of $p$ gaps, matching
bounds give the factor $\exp(\Theta(\log n/p^{1/r}))$ beyond linear cost.

The reason is simple.  Few products force the smallest base $m$ to be large,
and fullness makes $m$ the next-to-last gap.  An input built from reversed
blocks is already sorted for every earlier gap, forcing $\Omega(nm)$ work in
the final pass.  Powers of distinct primes give a matching construction.

For arbitrary gaps, we count current-gap multiples that earlier gaps cannot
form.  This gives upper and lower bounds for individual passes.  A Fourier
argument gives necessary conditions for small total cost, while short
nonnegative sums give sufficient conditions.  In both settings, useful
distances must be available before they are needed.
\end{abstract}


\end{frontmatter}

\section{Introduction}

Shellsort is a comparison sorting algorithm whose performance is determined by
the gaps used by its insertion-sort passes.

For a decreasing gap sequence
$H=(h_1>h_2>\cdots>h_p=1)$, the pass with gap $h_j$ sorts positions with the
same remainder modulo $h_j$.
Large gaps move keys over long distances; the final gap $1$ completes the
sort.  We write $W_n(H)$ and $C_n(H)$ for the worst-case numbers of exchanges
and comparisons on $n$ keys.  Since each exchange uses a comparison and each
pass has at most one additional linear term, lower bounds for $W_n$ also apply
to $C_n$, while an upper bound for $W_n$ becomes one for $C_n$ after adding
$O(np)$.

Pratt's sequence
\[
  \{2^u3^v<n:u,v\in\mathbb N_0\}
\]
uses $O(n\log^2 n)$ comparisons \cite{pratt1971}.  It is a full product grid:
all products below $n$ of the two generators are used as gaps.  This suggests
an immediate extension.  If two generators work well, can three or four
independent generators provide more intermediate gaps and make Shellsort
faster?

Why study this restricted family?  General bounds tell us what some gap
sequence might achieve, but they give little guidance for choosing the gaps.
Full product grids are among the few systematic families with provable
guarantees, and Pratt's sequence is their classical example.  Adding another
base number is the most direct way to create more intermediate gap sizes.  If
this worked, it would give a simple route to faster Shellsort; if it fails, a
different design principle is needed.

The question is not answered by existing bounds.  Bounds for arbitrary
$p$-pass sequences depend mainly on the number of passes and treat two-base
and three-base product constructions alike.  Merely counting the products is
also not enough: we must show that the forced gaps cause actual work on one
input.  To identify the right asymptotic cost, we must also construct a family
that meets the lower bound under the same gap budget.

Our answer is no for balanced full product grids.  For generators
$A=\{a_1,\ldots,a_r\}$, such a grid contains every product
$\prod_i a_i^{u_i}<n$.  Here ``full'' means that no such product is omitted,
``independent'' means that different exponent choices give different
products, and ``balanced'' means that every generator is at most a fixed power
of the smallest.

\paragraph{Main result in plain terms}
For each fixed $r\ge2$, among balanced full product grids with a budget of $p$
gaps, the factor beyond linear cost is
\[
  \exp\!\left(\Theta\!\left(\frac{\log n}{p^{1/r}}\right)\right)
\]
in the range where our lower and upper bounds match.  In particular, for
$p=\Theta(\log n)$,
\[
  \begin{aligned}
    r=2 &: \quad n\exp(\Theta(\sqrt{\log n})),\\
    r=3 &: \quad n\exp(\Theta((\log n)^{2/3})),\\
    r\ge2 &: \quad n\exp(\Theta((\log n)^{1-1/r})).
  \end{aligned}
\]
Each additional independent generator increases the exponent.  Thus the
extra intermediate gaps do not compensate for the larger gaps that the
product budget forces, and two generators are optimal in this class.  The
precise assumptions and full pass-budget statement appear in
\cref{thm:balanced-grid-barrier}.

\paragraph{Why more generators hurt}
The lower bound is the composition of two elementary facts:
\[
  \begin{gathered}
    \text{few products}
    \Longrightarrow \text{large smallest generator}\\
    \Longrightarrow \text{large next-to-last gap}
    \Longrightarrow \text{expensive final pass}.
  \end{gathered}
\]
Indeed, if only $p$ products lie below $n$, counting the possible exponent
choices and using balance give
$\log m=\Omega(\log n/p^{1/r})$ for the smallest generator $m$.  Because the
grid is full, $m$ is the next-to-last gap.  Reverse the keys inside consecutive
blocks of length $m$ and give later blocks larger values.  The input is already
sorted at every gap at least $m$, so only the final pass works; it must remove
$\Omega(nm)$ inversions.  For the upper bound, suitably sized powers of
distinct primes provide $\Theta(p)$ gaps and make enough distances available
early to match the lower bound up to constants in the exponent.

\paragraph{What survives beyond product grids}
The product-grid proof uses a more general fact: once a distance is a
nonnegative sum of earlier gaps, it remains sorted.  For an arbitrary gap
sequence, we therefore count the multiples of the current gap that cannot yet
be formed in this way.  \Cref{sec:directed} gives an upper bound from this
count, lower bounds showing when the count is unavoidable, and a
two-generator case in which the bounds differ only by a constant factor.

The appendices study two different kinds of arithmetic information.  A signed
integer relation allows subtraction and is useful for proving that a low cost
is impossible.  A nonnegative sum uses only addition and describes a distance
that earlier passes have actually sorted.  A standard Fourier argument turns
the first kind into necessary conditions; short sums of the second kind give
sufficient conditions.  These results complement rather than improve the
unrestricted $p^{-1/2}$ tradeoff of Poonen and Plaxton--Suel
\cite{poonen1993,plaxtonsuel1997}.  Their purpose is to explain why a relation
helps only if it becomes available before the pass that needs it.

\paragraph{Notation}
$\mathbb N=\{1,2,\ldots\}$ and $\mathbb N_0=\{0,1,\ldots\}$; all logarithms
are natural.  Integer representations may use coefficients of either sign,
whereas nonnegative representations use coefficients in $\mathbb N_0$.

\section{Background and related work}
\label{sec:prior}

Shell introduced diminishing-increment sorting in 1959 \cite{shell1959}.
Pratt's two-generator grid gives the classical $O(n\log^2 n)$ comparison bound
\cite{pratt1971}.  Sedgewick obtained an $O(n^{4/3})$ bound, and
Incerpi--Sedgewick obtained
$n^{1+O(1/\sqrt{\log n})}$ with $O(\log n)$ gaps
\cite{sedgewick1986,incerpi1985}.

For arbitrary sequences with $p$ passes, Poonen proved a worst-case lower
bound of the form
\[
  \Omega\!\left(n^{1+c/\sqrt p}\right)
\]
and a matching upper dependence
$O(pn^{1+O(1/\sqrt p)})$ \cite{poonen1993}.  Plaxton and Suel sharpened the
lower side and, together with Poonen, established the universal bound
\[
  \Omega\!\left(
    n\left(\frac{\log n}{\log\log n}\right)^2
  \right)
\]
\cite{plaxtonsuel1997}.  These results describe what is possible when the gaps
are unrestricted.  Because they depend on $p$ rather than on how the gaps are
generated, they treat product grids with two or three base numbers in the same
way.  We do not improve this universal tradeoff.

Pratt's gaps are products of $2$ and $3$.  Incerpi--Sedgewick's construction
likewise uses arithmetic coverage: once an array is sorted at distances $d$
and $e$, it is sorted at every nonnegative combination of $d$ and $e$
\cite{incerpi1985}.  Product grids make these useful multiples available at
successive scales.

Our question is narrower: what is the best cost when all products of $r$
balanced, independent base numbers must be used?  The exponent now depends on
$p$ as $p^{-1/r}$.  Two base numbers recover the classical $p^{-1/2}$
dependence, while every fixed $r>2$ is worse inside this family.  For fixed
generators, the same counting viewpoint gives
$\Theta_A(n(\log n)^{d(A)})$ comparisons, where $d(A)$ counts the independent
product directions.

For general sequences, classical upper bounds ask which distances are
nonnegative combinations of gaps already processed
\cite{incerpi1985,rosalesgarcia2009}.  Our bounds keep the order of the passes,
restrict attention to distances that occur in an array of length $n$, and
weight each distance by its number of position pairs.  Poonen's zero--one
construction supplies the corresponding hard-input idea \cite{poonen1993}.
The appendix Fourier bound is related to Zang's recent method \cite{zang2026};
we keep that complementary machinery outside the main argument.

\section{Generators versus passes in full product grids}
\label{sec:applications}

Let $A=\{a_1,\ldots,a_r\}\subseteq\{2,3,\ldots,n-1\}$, and define the full
product grid
\[
  H_A(n)=
  \left\{\prod_{i=1}^r a_i^{u_i}<n:
  u_i\in\mathbb N_0\right\},
\]
with duplicate values removed and the gaps used in decreasing order.  Write
$p_A(n)=|H_A(n)|$.  The generators are multiplicatively independent when
distinct exponent vectors give distinct products.  Put
\[
  m=\min_i a_i,
  \qquad R=\max_i a_i.
\]
For a fixed $K>1$, we call the grid balanced when $R\le m^K$; equivalently,
the generator logarithms are comparable up to a fixed factor.  When
$\gcd(A)=1$, write
\[
  \langle A\rangle
  =\left\{\sum_{i=1}^r x_i a_i:x_i\in\mathbb N_0\right\},
  \qquad
  g(\langle A\rangle)=|\mathbb N\setminus\langle A\rangle|.
\]
Thus $g(\langle A\rangle)$ is simply the number of positive integers that
cannot be written as a nonnegative sum of the generators.

When the lower and upper bounds meet, the theorem below has a simple meaning:
within this family, the best worst-case cost is linear cost multiplied by
$\exp(\Theta_{r,K}(\log n/p^{1/r}))$.  The full statement also records the two
bounds outside that matching range.

\begin{theorem}[Main theorem: generators versus passes]
\label{thm:balanced-grid-barrier}
Fix $r\ge2$ and $K>1$.  Let $p=p(n)$ be integer-valued, with
\[
  p\longrightarrow\infty,
  \qquad
  p=o((\log n)^r).
\]
Every multiplicatively independent set $A_n$ of $r$ generators satisfying
$\gcd(A_n)=1$, $R_n\le m_n^K$, and $|H_{A_n}(n)|\le p$ obeys
\[
  W_n(H_{A_n}),\ C_n(H_{A_n})
  \ge
  n\exp\!\left(
    \Omega_{r,K}\!\left(\frac{\log n}{p^{1/r}}\right)
  \right).
\]
Conversely, there are pairwise coprime, multiplicatively independent sets
$A_n$ with $R_n\le m_n^K$ such that
\[
  |H_{A_n}(n)|=\Theta_r(p),
  \qquad |H_{A_n}(n)|\le p,
\]
and
\[
  C_n(H_{A_n})
  \le np\exp\!\left(
    O_r\!\left(\frac{\log n}{p^{1/r}}\right)
  \right).
\]
If, in addition,
\[
  \log p=o\!\left(\frac{\log n}{p^{1/r}}\right),
\]
then the smallest possible worst-case exchange and comparison counts in this
class both have scale
\[
  n\exp\!\left(
    \Theta_{r,K}\!\left(\frac{\log n}{p^{1/r}}\right)
  \right).
\]
\end{theorem}

The most transparent consequence occurs when the number of gaps is
logarithmic.

\begin{corollary}[Logarithmically many gaps]\label{cor:balanced-log-gaps}
For $p=\Theta(\log n)$, the optimal scale in \cref{thm:balanced-grid-barrier}
is
\[
  n\exp\!\left(\Theta((\log n)^{1-1/r})\right).
\]
In particular, the overhead for $r=2,3,4$ is respectively
\[
  \exp(\Theta(\sqrt{\log n})),\qquad
  \exp(\Theta((\log n)^{2/3})),\qquad
  \exp(\Theta((\log n)^{3/4})).
\]
Thus two generators are optimal among balanced full product grids with any
fixed number of independent generators.
\end{corollary}

The proof has one short lower-bound chain and one matching construction:
\[
  \begin{gathered}
    \text{few products}
    \Longrightarrow \text{large smallest generator}\\
    \Longrightarrow \text{large next-to-last gap}
    \Longrightarrow \text{expensive final pass}.
  \end{gathered}
\]
The next two lemmas make the two implications quantitative.  Balanced
prime-power generators then give the upper bound.

\subsection{The two lower-bound lemmas}

\begin{lemma}[Few products force a large generator]
\label{lem:few-products-large-generator}
Let $A$ contain $r$ multiplicatively independent generators, let
$m=\min A$ and $R=\max A$, and suppose that $R\le m^K$.  Then
\[
  \log m
  \ge
  \frac{\log n}{4K(r!)^{1/r}p_A(n)^{1/r}}.
\]
\end{lemma}

\begin{proof}
Put $\mathcal L=\log n$ and
$t=\lfloor\mathcal L/(2\log R)\rfloor$.  Multiplicative independence makes
all products with $u_1+\cdots+u_r\le t$ distinct and smaller than $n$, so
\[
  p_A(n)\ge\binom{t+r}{r}.
\]
If $\log R>\mathcal L/4$, the claimed bound follows from
$\log m\ge(\log R)/K$ and $p_A(n)\ge1$.  Otherwise
$t\ge\mathcal L/(4\log R)$, and hence
\[
  p_A(n)
  \ge\frac1{r!}
  \left(\frac{\mathcal L}{4\log R}\right)^r.
\]
Rearranging and using $\log m\ge(\log R)/K$ proves the result.
\end{proof}

\begin{lemma}[A large next-to-last gap is expensive]
\label{lem:penultimate-gap}
Let
\[
  H=(h_1>\cdots>h_{p-1}>h_p=1),\qquad p\ge2,
\]
and put $b=h_{p-1}$.  Write $n=qb+s$, where $q\ge1$ and $0\le s<b$.
Then
\[
  W_n(H)\ge q\binom b2+\binom s2
  \ge \frac{n(b-1)}4.
\]
The same lower bound holds for $C_n(H)$.
\end{lemma}

\begin{proof}
Partition the positions into $q$ blocks of length $b$ and one final block of
length $s$.  Give later blocks larger values, but reverse the values within
each block.  If $h\ge b$, positions $i$ and $i+h$ lie in different blocks, so
the input is already $h$-sorted.  Every nonfinal pass therefore makes no
exchange.  The final pass removes exactly
$q\binom b2+\binom s2$ inversions.  Since $n-s=qb\ge n/2$, we have
$q\binom b2=(n-s)(b-1)/2\ge n(b-1)/4$.  Every exchange requires a comparison.
\end{proof}

\subsection{The matching upper construction}

Call an array $d$-sorted when $x_i\le x_{i+d}$ whenever both positions exist.
We use the following standard closure property.

\begin{lemma}[Sorted distances persist]\label{lem:persistence}
For the usual insertion implementation of a Shellsort pass:
\begin{enumerate}[label=(\roman*)]
  \item $h$-sorting a $d$-sorted array preserves $d$-sortedness;
  \item an array that is $d$- and $e$-sorted is $(ad+be)$-sorted for all
  $a,b\in\mathbb N_0$.
\end{enumerate}
\end{lemma}

The proof is given in \ref{app:directed-details}.  It is the familiar reason
that nonnegative combinations of earlier gaps can be ignored by later passes.

\begin{proposition}[Direct upper bound for a full product grid]
\label{prop:full-parent-grid}
If $\gcd(A)=1$, then
\[
  W_n(H_A)
  \le n p_A(n)g(\langle A\rangle),
  \qquad
  C_n(H_A)
  \le n p_A(n)\bigl(1+g(\langle A\rangle)\bigr).
\]
\end{proposition}

\begin{proof}
Consider a pass with gap $h$.  If
$k=\sum_i x_i a_i$ and $kh<n$, then every $a_i h$ with $x_i>0$ is a larger
grid gap and has already been processed.  By \cref{lem:persistence}, the array
is therefore $kh$-sorted.  Only multipliers
$k\notin\langle A\rangle$ can contribute inversions in the $h$-chains.  There
are at most $g(\langle A\rangle)$ such multipliers, each accounting for fewer
than $n$ position pairs.  This bounds the exchanges in each pass.  Summing over
the passes gives the first inequality; insertion sort uses at most one further
comparison per key and pass, giving the second.
\end{proof}

\begin{proof}[Proof of \cref{thm:balanced-grid-barrier}]
For the lower bound, put $p_n=|H_{A_n}(n)|$.  Then
\cref{lem:few-products-large-generator} gives
\[
  m_n
  \ge
  \exp\!\left(
    \frac{\log n}{4K(r!)^{1/r}p_n^{1/r}}
  \right).
\]
The smallest nonunit member of a full grid is $m_n$, so it is the next-to-last
gap.  Since $p_n\le p$, \cref{lem:penultimate-gap} gives the stated lower bound.

For the upper bound, fix distinct primes $\pi_1,\ldots,\pi_r$ and put
\[
  y_n=D\frac{\log n}{p^{1/r}},
  \qquad
  a_{i,n}=\pi_i^{\lfloor y_n/\log\pi_i\rfloor},
\]
where $D$ is a sufficiently large constant depending only on $r$.  The
assumptions on $p$ give $y_n\to\infty$ and $y_n=o(\log n)$.  These generators
are pairwise coprime and multiplicatively independent.  Moreover,
  $y_n-\log\pi_i<\log a_{i,n}\le y_n$, so $R_n\le m_n^K$ for all sufficiently
  large $n$.  Counting the feasible exponent tuples gives
\[
  |H_{A_n}(n)|
  =\Theta_r\!\left(\left(\frac{\log n}{y_n}\right)^r\right)
  =\Theta_r(p).
\]
Choosing $D$ large enough makes this count at most $p$.  Because the first two
generators are coprime, Sylvester's formula gives
\[
  g(\langle A_n\rangle)
  \le g(\langle a_{1,n},a_{2,n}\rangle)
  =\frac{(a_{1,n}-1)(a_{2,n}-1)}2
  \le\tfrac12e^{2y_n}.
\]
Now \cref{prop:full-parent-grid} gives the comparison upper bound.  Under the
last condition in the theorem, the factor $p$ is absorbed by the exponential,
so the upper and lower bounds have the same scale.
\end{proof}

The corollary follows by setting $p=\Theta(\log n)$ in
\cref{thm:balanced-grid-barrier}; its matching condition holds because
$\log\log n=o((\log n)^{1-1/r})$.

\noindent\textit{Classical consistency checks.}
For a fixed generator set with $d$ independent product directions, the same
argument gives $\Theta((\log n)^d)$ gaps and comparison cost
$\Theta(n(\log n)^d)$.  Two directions
recover Pratt's bound; allowing the generators to grow recovers the
Incerpi--Sedgewick scale.  Statements and proofs are in
\ref{app:classical-consequences}.

\FloatBarrier

\section{Bounds for individual passes beyond product grids}
\label{sec:directed}

The product-grid proof uses a broader principle: earlier gaps make their
nonnegative combinations sorted.  Consider a later pass with gap $h_j$.  A
multiple $mh_j<n$ cannot contain an inversion if it is a nonnegative sum of
earlier gaps; otherwise it may still contain any of the $n-mh_j$ position
pairs at that distance.  We now turn this observation into bounds for an
arbitrary decreasing sequence $H=(h_1>\cdots>h_p=1)$.

For $0\le s\le p$, let
\[
  S_s=\langle h_1,\ldots,h_s\rangle_{\mathbb N_0},
  \qquad S_0=\{0\}.
\]
For pass $j$, define
\[
  \begin{aligned}
    M_j&=\left\lfloor\frac{n-1}{h_j}\right\rfloor,\\
    G_j(n)&=\{1\le m\le M_j:mh_j\notin S_{j-1}\},\\
    D_j(n)&=\sum_{m\in G_j(n)}(n-mh_j).
  \end{aligned}
\]
Thus $D_j(n)$ counts the position pairs, at multiples of $h_j$, that need not
have been sorted before pass $j$.  We also use
\[
  U_s(n)=\sum_{\substack{1\le d<n\\d\notin S_s}}(n-d),
\]
the corresponding count over all distances after the first $s$ passes.

\subsection{An upper bound from already-sorted distances}

\begin{theorem}[Upper bound from missing multiples]
\label{thm:directed-upper}
Let $T_j(\sigma;H)$ be the number of exchanges in pass $j$.  For every input,
$T_j(\sigma;H)\le D_j(n)$.  Moreover,
\[
  \boxed{
  W_n(H)\le
  \mathfrak U_n(H):=
  \min_{0\le s\le p}
  \left(\sum_{j=1}^{s}D_j(n)+U_s(n)\right),}
\]
The usual gapped-insertion implementation also satisfies
$C_n(H)\le\mathfrak U_n(H)+np$.
\end{theorem}

\begin{proof}
By \cref{lem:persistence}, before pass $j$ the array is $d$-sorted for every
$d\in S_{j-1}$.  Hence only $m\in G_j(n)$ can contribute inverted pairs at
distance $mh_j$, giving $T_j\le D_j(n)$.  After pass $s$, only distances outside
$S_s$ can remain inverted, giving $U_s(n)$; minimize over $s$.  At most one
further comparison per key and pass gives the comparison estimate.
\end{proof}

An ordinary count would ask only how many positive integers cannot be formed.
Here $G_j(n)$ keeps the pass order and discards distances that do not fit in
the array, while $D_j(n)$ gives more weight to distances that occur in more
position pairs.

\subsection{Inputs that realize many missing pairs}

The upper count has a converse: one can keep every distance in $S_{j-1}$ sorted
while reversing many pairs counted by $D_j(n)$.  Averaging over the possible
remainders gives the following result.

\begin{theorem}[Lower bounds from missing distances]
\label{thm:directed-lower}
For every decreasing gap sequence,
\[
  \begin{aligned}
    W_n(H)&\ge
      \max\left\{
        D_1(n),\
        \max_{2\le j\le p}
          \left\lceil\frac{D_j(n)}{h_{j-1}}\right\rceil
      \right\},\\
    W_n(H)&\ge
      \frac{U_s(n)}{h_s(2h_{s+1}-1)}
      \qquad (1\le s<p).
  \end{aligned}
\]
Both bounds also hold for $C_n(H)$.
\end{theorem}

The factor $h_s$ comes from choosing one remainder class modulo $h_s$.  When the
prefix is generated by two coprime integers, a classical zero--one input uses
all remainder classes at once and removes that factor up to a constant.

\begin{theorem}[A constant-factor converse for two generators]
\label{thm:two-generator-cutoff-lower}
Suppose, for some $1\le s<p$, that $S_s=\langle a,b\rangle$, where $a,b\ge2$
are coprime, and that $n\ge(a-1)(b-1)$.  There is a permutation sorted at every
distance in $S_s$ with at least $U_s(n)/24$ inversions.  Consequently,
\[
  W_n(H)\ge
  \frac{U_s(n)}{24(2h_{s+1}-1)}.
\]
If $s=p-1$, this simplifies to $W_n(H)\ge U_{p-1}(n)/24$.  The same
conclusions hold for $C_n(H)$.
\end{theorem}

The proofs of both lower bounds are given in \ref{app:directed-details}.
They use the same idea as the main theorem: make the early passes do nothing,
then force a later pass to remove many chosen inversions.

\FloatBarrier

\section{Discussion and open directions}
\label{sec:discussion}

The main theorem settles the most direct extension of Pratt's construction.
If every product is retained and no base is more than a fixed power of the
smallest, adding a third or later generator does not help; it makes the
exponent in the worst-case cost larger.  This matters because it rules out a
simple and natural route to faster Shellsort: adding more intermediate product
gaps is not enough.

The scope is deliberately precise.  The result is not a stronger universal
lower bound for Shellsort.  Existing work already gives the $p^{-1/2}$
dependence for arbitrary $p$-pass sequences
\cite{poonen1993,plaxtonsuel1997}.  Our theorem instead identifies the extra
cost forced by using every product of $r$ independent base numbers.  Including
every product makes the smallest base the next-to-last gap; the size condition
prevents one base from carrying almost the entire scale.
Whether two bases remain best for sparse product sets, very unequal bases, or
product rules that change with $n$ remains open.

For general gap sequences, the broader lesson is about timing.  An arithmetic
relation helps a later pass only if earlier gaps have already made the needed
distance sorted.  In concrete terms, the distance must be obtainable by adding
earlier gaps, without subtraction, before the pass begins.  The bounds in
\cref{sec:directed} formalize this point.  \ref{sec:signed} and
\ref{sec:profiles} give complementary lower bounds, while
\ref{sec:separation} gives sufficient conditions for a small cost.  The
examples in \ref{app:additional-consequences} show why information about all
gaps at once can overestimate the work of one particular pass.

This leads to a plain question: can every gap sequence using
$O(n\log^2 n)$ comparisons be explained by this same local mechanism?  More
precisely, before almost every pass, do short sums of earlier gaps cover the
remainder classes that the pass needs?  The sufficient condition in
\cref{cor:apery-runtime} has this form.  A proof would move the general upper
bound close to a characterization; a counterexample would reveal a genuinely
different way to design efficient gaps.

\section*{CRediT Authorship Contribution Statement}

\textbf{Ziqi Zhao:} Conceptualization, Formal analysis, Software, Writing --
original draft, Writing -- review \& editing.

\textbf{Qingjian Ni:} Software, Supervision, Validation, Writing -- review \&
editing.

\section*{Funding}

This work was supported by the National Natural Science Foundation of China
(Grant No.~12273003).

\section*{Declaration of Competing Interest}

The authors declare that they have no known competing financial interests or
personal relationships that could have appeared to influence the work
reported in this paper.

\section*{Data and Code Availability}

The analytical results in this article do not rely on experimental data.  The
computational scripts used for regression checks and quantitative calibration
are summarized in \ref{app:verification} and available in the public
repository:
\begin{center}
  \url{https://github.com/qisumi/chronological-shellsort-artifact}.
\end{center}

\appendix
\renewcommand{\thetheorem}{\Alph{section}.\arabic{theorem}}
\renewcommand{\theproposition}{\Alph{section}.\arabic{proposition}}
\renewcommand{\thelemma}{\Alph{section}.\arabic{lemma}}
\renewcommand{\thecorollary}{\Alph{section}.\arabic{corollary}}
\renewcommand{\thedefinition}{\Alph{section}.\arabic{definition}}
\renewcommand{\theremark}{\Alph{section}.\arabic{remark}}

The appendices have three roles.  \ref{app:classical-consequences} connects
the main theorem with classical product grids, and
\ref{app:directed-details} supplies the deferred proofs of the general
pass-specific bounds.  \ref{sec:signed}, \ref{sec:profiles}, and
\ref{sec:separation} develop the complementary lower and upper conditions for
arbitrary gap sequences; \ref{app:additional-consequences} gives two examples
that mark the limits of those conditions.  The final appendix records the
scope of the computational checks.

\section{Classical product-grid consequences}
\label{app:classical-consequences}

For $a>1$, let $v(a)$ be its vector of prime valuations and set
\[
  d(A)=\dim_{\mathbb Q}\operatorname{span}_{\mathbb Q}
  \{v(a):a\in A\}.
\]

\begin{proposition}[Fixed-generator product grids]
\label{prop:fixed-parent-barrier}
For every fixed finite $A\subseteq\{2,3,\ldots\}$ with $\gcd(A)=1$,
\[
  d(A)\ge2,
  \qquad
  p_A(n)=\Theta_A((\log n)^{d(A)}),
  \qquad
  C_n(H_A)=\Theta_A\bigl(n(\log n)^{d(A)}\bigr).
\]
Thus three independent fixed generators cost $\Theta(n\log^3 n)$ comparisons,
whereas two cost $\Theta(n\log^2 n)$.
\end{proposition}

\begin{proof}
The valuation vectors lie in a rank-$d(A)$ lattice inside a box of side
$O_A(\log n)$, giving the upper product count.  Products of $d(A)$ independent
members with bounded exponent sum give the matching lower count.  If $d(A)=1$,
all members of $A$ share a prime divisor, contrary to $\gcd(A)=1$.
\Cref{prop:full-parent-grid} gives the comparison upper bound.  For the lower
bound, $\Theta_A((\log n)^{d(A)})$ gaps lie below $n/2$, and each uses at least
$n/2$ comparisons even on sorted input.
\end{proof}

\begin{corollary}[Pratt's two-generator grids]
\label{cor:two-generator-grid}
For fixed coprime integers $a,b\ge2$,
\[
  |H_{a,b}(n)|
  =\frac{(\log n)^2}{2\log a\log b}+O(\log n),
  \qquad
  C_n(H_{a,b})=\Theta(n\log^2 n).
\]
Every pass uses at most $n(a-1)(b-1)/2$ exchanges.
\end{corollary}

\begin{proof}
Lattice-point counting in
$u\log a+v\log b<\log n$ gives the number of gaps.  Sylvester's formula gives
$g(\langle a,b\rangle)=(a-1)(b-1)/2$, so the proof of
\cref{prop:full-parent-grid} gives the per-pass exchange bound and the total
upper bound.  The matching comparison lower bound follows from the sorted
input as in \cref{prop:fixed-parent-barrier}.
\end{proof}

This recovers Pratt's $\Theta(n\log^2 n)$ bound.  Allowing the two generators
to grow with $n$ recovers the classical subpolynomial overhead.

\begin{corollary}[Two generators varying with $n$]
\label{cor:scaled-two-grid}
For every fixed $\varepsilon>0$ and all sufficiently large $n$, there are
coprime integers $a_n,b_n$ such that $H_{a_n,b_n}(n)$ has
$O_\varepsilon(\log n)$ gaps and
\[
  C_n(H_{a_n,b_n})
  =O_\varepsilon\!\left(n\exp(\varepsilon\sqrt{\log n})\right)
  =O_\varepsilon\!\left(n^{1+\varepsilon/\sqrt{\log n}}\right).
\]
\end{corollary}

\begin{proof}
Put $\mathcal L=\log n$, $c=\varepsilon/4$, and
$q=\lceil\exp(c\sqrt{\mathcal L})\rceil$; take $(a_n,b_n)=(q,q+1)$.
Lattice-point counting gives $p_{q,q+1}(n)=O_c(\mathcal L)$, while
Sylvester's formula gives
$g(\langle q,q+1\rangle)=\exp((2c+o(1))\sqrt{\mathcal L})$.
The result follows from \cref{prop:full-parent-grid}.
\end{proof}

This is the $n^{1+O(1/\sqrt{\log n})}$ scale of Incerpi and Sedgewick
\cite{incerpi1985}.  For arbitrary sequences the best pass dependence has
exponent scale $p^{-1/2}$ \cite{poonen1993,plaxtonsuel1997}.  Rank two matches
that exponent, whereas every fixed $r>2$ is worse inside balanced full product
grids.

\section{A Fourier lower bound for arbitrary gap sequences}
\label{sec:signed}

This appendix gives a lower bound for an arbitrary gap sequence.  It is
independent of the product-grid classification, but it uses the same ordering
principle: a relation among late gaps cannot help during an early pass.

For a prefix of the gaps, we choose a Fourier character that changes little
along every gap in that prefix.  A suitable input then has large Fourier
discrepancy.  The early passes can remove that discrepancy only slowly, and
the remaining passes move keys only by shorter distances.  The resulting
lower bound will be combined with elementary approximation and integer
representations in \ref{sec:profiles}.

\subsection{A Fourier lower bound for every prefix}

Let $\omega_n=\exp(2\pi i/n)$ and
\[
  P_\sigma(z)=\sum_{r=1}^{n}\sigma(r)z^r.
\]
For a prefix $H_s=(h_1,\ldots,h_s)$ and $1\le k<n$, define
\[
  \delta_s(k)=\max_{1\le j\le s}\normfrac{kh_j/n},
  \qquad
  \beta_s(H;n)=\min_{1\le k<n}\delta_s(k).
\]
Thus $\delta_s(k)$ is the largest Fourier error of character $k$ on the first
$s$ gaps, and $\beta_s$ is the smallest such error among the nonconstant
characters.  A rearrangement argument supplies a worst-case input for each
character.

\begin{lemma}[A rearrangement inequality]\label{lem:rearrangement}
If $x_1,\ldots,x_n\in\mathbb R$ and $\sum_i x_i=0$, then there is a permutation
$\pi$ of $[n]$ such that
\[
  \sum_{i=1}^{n}\pi(i)x_i\ge\frac n4\sum_{i=1}^{n}|x_i|.
\]
\end{lemma}

\begin{proof}
Sort the $x_i$ increasingly and assign the labels $1,\ldots,n$ in the same order.  Let $m$ be the number of positive $x_i$ and let
$A=\sum_{x_i>0}x_i=\frac12\sum_i|x_i|$.  By the rearrangement (equivalently, Chebyshev sum) inequality, the positive contribution is at least $A$ times the average of the largest $m$ labels, namely $(2n-m+1)/2$.  The absolute value of the negative contribution is at most $A$ times the average of the smallest $n-m$ labels, namely $(n-m+1)/2$.  Their difference is $An/2=(n/4)\sum_i|x_i|$.
\end{proof}

\begin{lemma}[Uniform Fourier discrepancy]\label{lem:adversary}
For every nontrivial $n$th root of unity $z$, there exists $\sigma\in S_n$ such that
\[
  |P_\sigma(z)-P_{\mathrm{id}}(z)|\ge\frac{n^2}{4\pi}.
\]
\end{lemma}

\begin{proof}
Because $z\ne1$ and $z^n=1$, $\sum_{r=1}^{n}z^r=0$.  Averaging over
$\theta\in[0,2\pi]$ gives
\[
  \frac1{2\pi}\int_0^{2\pi}
  \sum_{r=1}^{n}\left|\Re(\ee^{-i\theta}z^r)\right|\,d\theta
  =\frac{2n}{\pi}.
\]
Choose $\theta$ for which the sum is at least $2n/\pi$, and apply
\cref{lem:rearrangement} to $x_r=\Re(\ee^{-i\theta}z^r)$.  The resulting permutation satisfies
\[
  |P_\sigma(z)|\ge\Re(\ee^{-i\theta}P_\sigma(z))\ge\frac{n^2}{2\pi}.
\]
If $|P_{\mathrm{id}}(z)|\le n^2/(4\pi)$, the triangle inequality proves the claim.  Otherwise, for a uniformly random permutation $\Sigma$,
$\mathbb E P_\Sigma(z)=0$.  Jensen's inequality gives
\[
  \max_\sigma|P_\sigma(z)-P_{\mathrm{id}}(z)|
  \ge\left|\mathbb E(P_\Sigma(z)-P_{\mathrm{id}}(z))\right|
  =|P_{\mathrm{id}}(z)|>\frac{n^2}{4\pi}.
\]
\end{proof}

\begin{theorem}[Fourier lower bound for a prefix]\label{thm:fourier-lb}
For every standard decreasing Shellsort gap sequence
$H=(h_1>\cdots>h_p=1)$ and every $1\le s<p$,
\[
  \boxed{
  W_n(H)\ge
  \frac{n^2}{8\pi\bigl(h_{s+1}+\pi n\,\beta_s(H;n)\bigr)}.}
\]
Taking the maximum over $s$ gives
\[
  W_n(H)\ge
  \max_{1\le s<p}
  \frac{n^2}{8\pi\bigl(h_{s+1}+\pi n\,\beta_s(H;n)\bigr)}.
\]
\end{theorem}

\begin{proof}
Choose $1\le k<n$ attaining $\beta_s$, set $z=\omega_n^k$, and choose
$\sigma$ by \cref{lem:adversary}.  Run Shellsort on $\sigma$, and let $\rho$ be the permutation after the first $s$ passes.  Let $T_{\mathrm L}$ and $T_{\mathrm S}$ be the exchange counts in the first $s$ and remaining passes; thus $T=T_{\mathrm L}+T_{\mathrm S}$.

An exchange at distance $h_j$, $j\le s$, swapping values $u$ and $v$, changes the potential by
\[
  |u-v|\,|z^{h_j}-1|
  \le n\,2\pi\normfrac{kh_j/n}
  \le2\pi n\beta_s.
\]
Therefore
\[
  |P_\sigma(z)-P_\rho(z)|\le2\pi n\beta_s T_{\mathrm L}.
\]

The remaining exchanges have distance at most $h_{s+1}$.  For a permutation
$\pi$, write
\[
  F(\pi)=\sum_{r=1}^{n}|\pi(r)-r|.
\]
If $\pi'$ is obtained from $\pi$ by transposing positions $u<v$, then the
triangle inequality gives the Lipschitz estimate
\[
  F(\pi')-F(\pi)
  \le 2(v-u).
\]
Reverse the remaining exchanges starting from the identity.  Each such
transposition has $v-u\le h_{s+1}$, and therefore
\[
  \sum_{r=1}^{n}|\rho(r)-r|\le2h_{s+1}T_{\mathrm S}.
\]
Hence
\[
  |P_\rho(z)-P_{\mathrm{id}}(z)|
  \le\sum_{r=1}^{n}|\rho(r)-r|
  \le2h_{s+1}T_{\mathrm S}.
\]
Combining the two displays,
\[
  \frac{n^2}{4\pi}
  \le |P_\sigma(z)-P_{\mathrm{id}}(z)|
  \le2T\bigl(h_{s+1}+\pi n\beta_s\bigr),
\]
which proves the theorem.
\end{proof}

\begin{remark}[Why every prefix is considered]
A quantity computed from all gaps can be made uninformative by one small gap
and can therefore miss a lower bound that is present during early passes.  The
sequence $s\mapsto\beta_s$ respects the actual order of Shellsort: the first
$s$ passes are bounded through their Fourier error, while all later passes are
bounded only through their maximum movement length $h_{s+1}$.
The definition of $\beta_s$ and \cref{thm:fourier-lb} apply for every value of
$\gcd(n,h_1,\ldots,h_s)$, including the case $\beta_s=0$.
\end{remark}

\section{What the Fourier lower bound implies}
\label{sec:profiles}

We estimate $\beta_s$ in two ways.  Simultaneous approximation gives a bound
that holds for every gap sequence.  Integer relations among successive gaps
can give a stronger bound for a particular sequence.  Both estimates yield
necessary conditions for nearly linear worst-case cost.

\subsection{Universal simultaneous approximation}

Set
\[
  N_s=\left\lfloor(n-1)^{1/s}\right\rfloor.
\]

\begin{lemma}[Simultaneous modular approximation]\label{lem:dirichlet}
For integers $a_1,\ldots,a_s$, there exists $1\le k<n$ such that
\[
  \max_{1\le j\le s}\normfrac{ka_j/n}\le\frac1{N_s}.
\]
\end{lemma}

\begin{proof}
Partition $[0,1)^s$ into $N_s^s$ half-open cubes of side $1/N_s$.  Among the
$N_s^s+1$ points
\[
  \left(\left\{\frac{ra_1}{n}\right\},\ldots,
  \left\{\frac{ra_s}{n}\right\}\right),
  \qquad r=0,1,\ldots,N_s^s,
\]
two lie in the same cube.  Their index difference $k$ satisfies
$1\le k\le N_s^s\le n-1$, and each torus-coordinate difference is at most $1/N_s$.
\end{proof}

\begin{corollary}[Universal bound from simultaneous approximation]\label{cor:dirichlet-profile}
For every $1\le s<p$,
\[
  \beta_s(H;n)\le\frac1{N_s},
\]
and hence
\[
  \boxed{
  W_n(H)\ge
  \max_{1\le s<p}
  \frac{n^2}{8\pi\left(h_{s+1}+\pi n/N_s\right)}.}
\]
For fixed $s$ and $h_{s+1}=O(n^{1-1/s})$, this gives
$W_n(H)=\Omega(n^{1+1/s})$.
\end{corollary}

\begin{proof}
Apply \cref{lem:dirichlet} to $(h_1,\ldots,h_s)$ and substitute the resulting bound on $\beta_s$ into \cref{thm:fourier-lb}.
\end{proof}

This corollary converts the standard $n^{-1/s}$ simultaneous-approximation
estimate directly into a Shellsort lower bound.

\subsection{Using integer relations between successive gaps}

The universal estimate treats every gap in the prefix independently.  If the
current gap has a short integer representation using earlier gaps, then a
Fourier character with small error on the earlier gaps also has small error on
the current one.  For $j\ge2$, define the minimum coefficient length
\[
  \lambda_j(H)=
  \min\left\{
    \|x\|_1:x\in\mathbb Z^{j-1},
    \sum_{i<j}x_i h_i=h_j
  \right\},
\]
with $\min\varnothing=+\infty$.

\begin{lemma}[Effect of an integer representation]\label{lem:signed-transfer}
The quantity $\lambda_j(H)$ is finite if and only if
$\gcd(h_1,\ldots,h_{j-1})$ divides $h_j$, and whenever it is finite one has
$\lambda_j(H)\ge2$.  Moreover, for every character index $k$,
\[
  \boxed{
  \normfrac{k h_j/n}
  \le \lambda_j(H)\,\delta_{j-1}(k).}
\]
\end{lemma}

\begin{proof}
The finiteness criterion is B\'ezout's identity.  Strict decrease and
positivity of the gaps exclude integer representations of $h_j$ having
$\ell_1$ length zero or one.  For a minimizing vector $x$, symmetry and
subadditivity of the torus norm give
\[
  \normfrac{k h_j/n}
  =\normfrac{\sum_{i<j}x_i k h_i/n}
  \le\sum_{i<j}|x_i|\normfrac{k h_i/n}
  \le\lambda_j(H)\,\delta_{j-1}(k).
\]
\end{proof}

Choosing a character attaining $\beta_{j-1}$ in
\cref{lem:signed-transfer} gives the propagation step
\[
  \beta_j(H;n)\le\lambda_j(H)\,\beta_{j-1}(H;n).
\]

\begin{proposition}[Combining simultaneous approximation and integer relations]
\label{prop:hybrid-profile}
Let $2\le j_0\le s<p$ and suppose
$\lambda_j(H)<\infty$ for $j_0<j\le s$.  Put
\[
  \Pi_{j_0,s}(H)=\prod_{j=j_0+1}^{s}\lambda_j(H),
  \qquad
  \mathcal B_{j_0,s}(H;n)
  =\min\left\{\frac12,\frac{\Pi_{j_0,s}(H)}{N_{j_0}}\right\},
\]
where an empty product is one.  Then
\[
  \beta_s(H;n)\le\mathcal B_{j_0,s}(H;n)
\]
and hence
\[
  \boxed{
  W_n(H)\ge
  \frac{n^2}{8\pi\left(
    h_{s+1}+\pi n\,\mathcal B_{j_0,s}(H;n)
  \right)}.}
\]
\end{proposition}

\begin{proof}
By \cref{cor:dirichlet-profile},
$\beta_{j_0}\le1/N_{j_0}$.  Iterating the inequality in
\cref{lem:signed-transfer} gives
$\beta_s\le\Pi_{j_0,s}/N_{j_0}$, while the definition of the torus norm
always gives $\beta_s\le1/2$.  Substitute the resulting upper bound on
$\beta_s$ into \cref{thm:fourier-lb}.
\end{proof}

Thus simultaneous approximation need only be applied to a short prefix; the
remaining Fourier errors can be bounded using integer relations.  In
particular, if
$\lambda_j\le\Lambda$ along the chain, then
$\beta_s\le\min\{1/2,\Lambda^{s-j_0}/N_{j_0}\}$.  For the Pratt grid,
whenever $3h_j<n$ the earlier gaps contain both $2h_j$ and $3h_j$, so
$h_j=3h_j-2h_j$ and $\lambda_j=2$.

\subsection{How many large gaps are necessary}

The universal Fourier estimate gives a necessary condition on the large gaps
of every Shellsort family with nearly linear cost.

\begin{theorem}[Nearly linear cost requires many large gaps]\label{thm:near-linear-density}
Let $\{H^{(n)}\}$ be a family of decreasing gap sequences, with
\[
  H^{(n)}=(h_{n,1}>\cdots>h_{n,p_n}=1).
\]
Suppose that for fixed $c>0$,
\[
  W_n(H^{(n)})=O(n(\log n)^c).
\]
For every fixed $\varepsilon>0$, set
\[
  s_n=\left\lfloor
  \frac{\log n}{(c+\varepsilon)\log\log n}
  \right\rfloor.
\]
Then, for all sufficiently large $n$,
\[
  p_n\ge s_n+1
  \qquad\text{and}\qquad
  h_{n,s_n+1}=\Omega\!\left(\frac{n}{(\log n)^c}\right).
\]
Thus any such family must contain
$\Omega(\log n/\log\log n)$ gaps of order at least $n/(\log n)^c$.
\end{theorem}

\begin{proof}
Let $A$ be such that $W_n(H^{(n)})\le A n(\log n)^c$ for all large $n$.  First suppose
$p_n\le s_n$.  If $p_n=1$, Shellsort is insertion sort and has quadratic worst case, a contradiction.  Otherwise apply \cref{cor:dirichlet-profile} with $s=p_n-1$ and $h_{n,p_n}=1$.  Since $p_n-1<s_n$, its $N_s$ is at least $N_{s_n}$, while
\[
  N_{s_n}=(\log n)^{c+\varepsilon+o(1)}.
\]
In particular, $N_{s_n}\ge(\log n)^{c+\varepsilon/2}$ for all sufficiently
large $n$, so the $o(1)$ term cannot consume the fixed $\varepsilon$ margin.
The resulting denominator is
$1+O(n/(\log n)^{c+\varepsilon+o(1)})=o(n/(\log n)^c)$, contradicting the assumed upper bound.  Hence $p_n\ge s_n+1$.

Now apply \cref{cor:dirichlet-profile} with $s=s_n$:
\[
  A n(\log n)^c
  \ge \frac{n^2}{8\pi(h_{n,s_n+1}+\pi n/N_{s_n})}.
\]
Rearranging,
\[
  h_{n,s_n+1}+\frac{\pi n}{N_{s_n}}
  \ge\frac{n}{8\pi A(\log n)^c}.
\]
The second term is $o(n/(\log n)^c)$, leaving the claimed lower bound on
$h_{n,s_n+1}$.
\end{proof}

\begin{remark}[Extension using integer relations]
The same argument applies to substantially longer prefixes when successive
gaps have short integer representations.  More precisely, suppose
$W_n(H^{(n)})\le A n(\log n)^c$ for all sufficiently large $n$.  Fix
$j_0\ge2$, $\Lambda\ge2$, and $\eta>0$, and choose
$j_0\le s_n<p_n$ so that
\[
  \lambda_j(H^{(n)})\le\Lambda
  \quad (j_0<j\le s_n),
  \qquad
  s_n-j_0\le
  \frac{\log N_{j_0}-(c+\eta)\log\log n}{\log\Lambda}.
\]
Then \cref{prop:hybrid-profile} gives
\[
  \beta_{s_n}(H^{(n)};n)
  \le \frac{\Lambda^{s_n-j_0}}{N_{j_0}}
  \le \frac1{(\log n)^{c+\eta}}.
\]
Combining this estimate with the assumed upper bound and rearranging the
lower bound in \cref{prop:hybrid-profile} yields
\[
  h_{n,s_n+1}
  +O\!\left(\frac{n}{(\log n)^{c+\eta}}\right)
  \ge \frac{n}{8\pi A(\log n)^c},
\]
and hence
\[
  h_{n,s_n+1}=\Omega\!\left(\frac{n}{(\log n)^c}\right).
\]
For fixed $j_0$ and $\Lambda$, the permitted number of successive
representations is
\[
  \frac{\log n}{j_0\log\Lambda}-O(\log\log n)=\Theta(\log n).
\]
\end{remark}

Small $\beta_s$, slow reachability by integer combinations, and short
representations of successive gaps can each force a large exchange count.
The next section instead studies which distances have already been sorted by
nonnegative combinations of earlier gaps.

\section{Proofs and refinements for the pass-specific bounds}
\label{app:directed-details}

\subsection{Why sorted distances persist}

\begin{proof}[Proof of \cref{lem:persistence}]
For part (i), threshold the keys at an arbitrary real value.  Thresholding
commutes with sorting each $h$-chain, so it suffices to consider a zero--one
array.  Index positions from zero.  Suppose
$i=r+qh$ and $i+d=r'+q'h$, and put $c=q'-q\ge0$.  In the $h$-chain containing
$i+d$, at most the first $c$ positions have no predecessor at distance $d$ in
the chain containing $i$.  Because the input is $d$-sorted, every zero outside
those first $c$ positions maps backward to a zero.  If the $h$-sorted output
has a zero at rank $q'$ in the target chain, that chain originally contained
at least $q'+1$ zeros.  The source chain therefore contained at least
$(q'+1)-c=q+1$ zeros, so its output has a zero at rank $q$.
This proves preservation for every threshold and hence for arbitrary distinct
keys.

For part (ii), $x_i\le x_{i+d}\le x_{i+d+e}$ proves
$(d+e)$-sortedness; induction gives every nonnegative combination.  Applying
part (i) after each pass completes the proof.
\end{proof}

\subsection{Proof of the general lower bounds}

\begin{lemma}[Forcing inversions from one residue class]
\label{lem:residue-reversal}
Fix $1\le s<p$ and a residue $r\in\{0,\ldots,h_s-1\}$.  There is a permutation
$\sigma$ that is $d$-sorted for every $d\in S_s$ and has the following
property: whenever
\[
  0\le u<v<n,
  \qquad u\equiv r\pmod{h_s},
  \qquad v-u\notin S_s,
\]
the pair $(u,v)$ is an inversion of $\sigma$.
\end{lemma}

\begin{proof}
Put a precedence constraint $u\prec v$ whenever $u<v$ and $v-u\in S_s$.
The positions congruent to $r$ modulo $h_s$ are totally ordered by these
constraints because every positive multiple of $h_s$ belongs to $S_s$.  Call
this set of positions $C$.

For every $u\in C$ and every position $v$ incomparable with $u$, add the
constraint $v\prec u$.  The enlarged constraints have no directed cycle.
Indeed, suppose a cycle contains successive added edges $v_i\to u_i$, with
original constraints leading from $u_i$ to $v_{i+1}$.  Both $u_i$ and
$u_{i+1}$ lie in the totally ordered set $C$.  If
$u_{i+1}\preceq u_i$, then
$u_{i+1}\preceq u_i\preceq v_{i+1}$, contradicting the fact that
$u_{i+1}$ and $v_{i+1}$ are incomparable.  Thus $u_i\prec u_{i+1}$ for every
pair of successive added edges, which is impossible around a cycle.

List the positions in any order respecting all constraints and assign the
values $1,\ldots,n$ in that order.  The original constraints make the
permutation $d$-sorted for every $d\in S_s$.  For every pair in the statement,
the added constraint puts $v$ before $u$, so $\sigma(u)>\sigma(v)$.
\end{proof}

\begin{proof}[Proof of \cref{thm:directed-lower}]
For $j=1$, reverse the values inside every $h_1$-chain.  The first pass then
performs one exchange for every pair in such a chain, exactly $D_1(n)$ in
total.

Fix $2\le j\le p$.  The pairs counted by $D_j(n)$ have the form
\[
  (u,u+mh_j),
  \qquad mh_j\notin S_{j-1}.
\]
Group them by the residue of $u$ modulo $h_{j-1}$.  Some residue contains at
least $\lceil D_j(n)/h_{j-1}\rceil$ pairs.  Apply
\cref{lem:residue-reversal} with $s=j-1$ and this residue.  The resulting input
is already sorted at every earlier gap, so the first $j-1$ passes make no
exchanges.  Every selected pair is inverted and lies in one $h_j$-chain.
Insertion sort makes one exchange per inversion within those chains, proving
the first bound.

For the second bound, group all $U_s(n)$ pairs by the residue of their earlier
position modulo $h_s$.  By \cref{lem:residue-reversal}, one input is sorted at
every distance in $S_s$ and has at least $U_s(n)/h_s$ inversions.  The first
$s$ passes do nothing.  Every later exchange has distance at most $h_{s+1}$,
and an exchange at distance $d$ can decrease the global inversion number by at
most $2d-1$.  At least
\[
  \frac{U_s(n)/h_s}{2h_{s+1}-1}
\]
later exchanges are therefore necessary.  Every exchange uses a comparison.
\end{proof}

\subsection{Proof of the two-generator constant-factor bound}

\begin{proof}[Proof of \cref{thm:two-generator-cutoff-lower}]
Put
\[
  N=(a-1)(b-1),\qquad S=\langle a,b\rangle.
\]
The positive integers outside $S$ are all below $N$ and number $g=N/2$.
On the positions $0,\ldots,N-1$, form the zero--one word
\[
  y_t=\begin{cases}
    1,&t\in S,\\
    0,&t\notin S.
  \end{cases}
\]
This word is $d$-sorted for every $d\in S$: if $y_t=1$, then
$t+d\in S$ and hence $y_{t+d}=1$.  This is the two-generator instance of
the zero--one construction used in classical Shellsort lower bounds
\cite{poonen1993}.

Let the gaps of $S$ be $q_1<\cdots<q_g$.  Before $q_k$ there are
$q_k-(k-1)$ elements of $S$, including zero.  Hence the number $I$ of
inversions in the word is
\[
  I
  =\sum_{k=1}^{g}\bigl(q_k-(k-1)\bigr)
  =\sum_{k=1}^{g}q_k-\frac{g(g-1)}2.
\]
The classical formulas for the number and sum of these gaps give
\cite{brownshiue1993}
\[
  g=\frac N2,\qquad
  \sum_{k=1}^{g}q_k
  =\frac{N(2ab-a-b-1)}{12}.
\]
Substitution yields the exact count
\[
  I=\frac{(a^2-1)(b^2-1)}{24}
   =\frac{N(a+1)(b+1)}{24}.
\]

Repeat the word in consecutive blocks of length $N$.  More explicitly, for
$0\le t<n$ set
\[
  z_t=2\left\lfloor\frac tN\right\rfloor+y_{t\bmod N},
\]
using the corresponding prefix of the word in the final partial block.
All values in a later block exceed all values in an earlier block.  Thus
$(z_t)$ remains sorted at every distance in $S$: within a block this follows
from the zero--one sequence, and across blocks it follows from the increasing
value ranges.  Replace the $z_t$ by their ranks, breaking ties from left to
right.  This gives a permutation with the same sorting property and preserves
every strict inversion.

Each complete block contributes $I$ inversions.  Since $n\ge N$,
\[
  \operatorname{Inv}(\sigma)
  \ge\left\lfloor\frac nN\right\rfloor I
  \ge\frac{n(a+1)(b+1)}{48}
  \ge\frac{nN}{48}.
\]
All gaps of $S$ are below $n$, so
\[
  U_s(n)
  =\sum_{d\notin S}(n-d)
  \le ng=\frac{nN}{2}.
\]
Therefore $\operatorname{Inv}(\sigma)\ge U_s(n)/24$.  The first $s$ passes do
nothing because the input is sorted at every distance in $S_s$.  As in the proof of
\cref{thm:directed-lower}, every later exchange removes at most
$2h_{s+1}-1$ inversions.  When $s=p-1$, the only remaining pass has gap one
and removes exactly one inversion per exchange.
\end{proof}

\subsection{Further refinements}

Let $g_j^{(n)}=|G_j(n)|$, the number of missing multipliers in the range
$1\le m\le M_j$.  The weights in $D_j(n)$ give a sharper estimate than the
unweighted bound $D_j(n)\le n g_j^{(n)}$.

\begin{proposition}[Bounds using only the number of missing multipliers]\label{prop:truncated-genus-refinement}
Let $h=h_j$, $M=\lfloor(n-1)/h\rfloor$, $g=g_j^{(n)}$, and
$\rho=n-Mh\in\{1,\ldots,h\}$.  Then
\[
  \boxed{
  \rho g+\frac{hg(g-1)}2
  \le D_j(n)
  \le ng-\frac{hg(g+1)}2.}
\]
If $g>0$, equality in the lower bound holds if and only if
$G_j(n)=\{M-g+1,\ldots,M\}$, while equality in the upper bound holds if and
only if $G_j(n)=\{1,\ldots,g\}$.  When $g=0$, both bounds are equalities.
Moreover,
\[
  g_j^{(n)}
  \le
  \frac{1+\sqrt{1+8D_j(n)/h_j}}{2}.
\]
\end{proposition}

\begin{proof}
Each $m\in G_j(n)$ contributes $n-mh$ to $D_j(n)$, so
\[
  D_j(n)=ng-h\sum_{m\in G_j(n)}m.
\]
The $g$ distinct positive integers in $G_j(n)$ lie in $[1,M]$, so their sum
is at least $1+\cdots+g=g(g+1)/2$, with equality exactly when
$G_j(n)=\{1,\ldots,g\}$.  This is the claimed upper bound.

For the lower bound, enumerate the missing multipliers as
$m_1<\cdots<m_g$.  Since they are distinct and $m_g\le M$, one has
$m_r\le M-g+r$.  Hence
\[
  n-m_rh\ge n-(M-g+r)h=\rho+(g-r)h.
\]
Summing over $r$ gives
\[
  D_j(n)\ge g\rho+\frac{hg(g-1)}{2}.
\]
For $g>0$, equality holds throughout if and only if
$m_r=M-g+r$ for every $r$, equivalently
$G_j(n)=\{M-g+1,\ldots,M\}$.

For the bound on $g$, drop the nonnegative term $\rho g$ from the lower
bound and rearrange:
\[
  g^2-g-\frac{2D_j(n)}{h}\le0.
\]
The positive root gives the displayed inequality.
\end{proof}

The moment bounds give a concrete sufficiency criterion.

\begin{corollary}[A sufficient condition based on missing multipliers]\label{cor:directed-genus}
If a gap family satisfies
\[
  p_n+\sum_{j=1}^{p_n}g_j^{(n)}=O((\log n)^c),
\]
then its worst-case exchange and comparison counts are $O(n(\log n)^c)$.
\end{corollary}

The Ap\'ery-set condition in \cref{cor:apery-runtime} strengthens this
counting criterion.

\subsection{Putting the lower and upper bounds together}

The Fourier lower bound and the two bounds from missing distances combine as
follows.

\begin{corollary}[Combined lower and upper bounds]\label{cor:sandwich}
For every decreasing gap sequence,
\[
  \begin{aligned}
  W_n(H)\ge\max\Bigg\{&
    \max_{1\le r<p}
      \frac{n^2}{8\pi(h_{r+1}+\pi n\beta_r(H;n))},\
    D_1(n),\\
    &\max_{2\le j\le p}\frac{D_j(n)}{h_{j-1}},\
    \max_{1\le s<p}
      \frac{U_s(n)}{h_s(2h_{s+1}-1)}
  \Bigg\},
  \end{aligned}
\]
while
\[
  W_n(H)\le\mathfrak U_n(H).
\]
\end{corollary}

The later appendix on signed and nonnegative representations relates these
parameters under an explicit condition on the finite range $1\le m\le M_j$.

\section{From integer relations to nonnegative representations}
\label{sec:separation}

Subtracting one nonnegative representation from another gives an integer
relation.  A Fourier character that has small error on the earlier gaps but
large error on the current gap therefore limits which multipliers can have
short nonnegative representations.  Conversely, if the represented
multipliers are sufficiently dense in an interval, two must be consecutive;
subtracting their representations then gives the current gap.  The converse
of this statement requires additional information about sizes and residue
classes.

For the appendix results in this section, write
\[
  Q_j=\{m\in\mathbb N_0:mh_j\in S_{j-1}\}.
\]
Thus $Q_j$ is the full set of multipliers represented by gaps available before
pass $j$; the main text uses only its finite complement $G_j(n)$.

\subsection{Fourier separation limits short nonnegative representations}

For pass $j$, let $\ell_j^\to(m)$ be the smallest number of earlier gaps
needed to represent $m h_j$ with nonnegative coefficients:
\[
  \ell_j^\to(m)
  =\min\left\{
    \|x\|_1:x\in\mathbb N_0^{j-1},\ 
    \sum_{i<j}x_i h_i=m h_j
  \right\},
\]
with value $+\infty$ if no representation exists and
$\ell_j^\to(0)=0$.  Set
\[
  Q_j^{\le L}=\{0\le m\le M_j:\ell_j^\to(m)\le L\}.
\]

\begin{theorem}[Fourier separation and representation length]\label{thm:short-ray-phase}
Fix $2\le j\le p$, $L\ge0$, and $1\le k<n$.  For all
$u,v\in Q_j^{\le L}$,
\[
  \normfrac{k(u-v)h_j/n}\le2L\,\delta_{j-1}(k).
\]
If, in addition,
\[
  \normfrac{k h_j/n}>2L\,\delta_{j-1}(k),
\]
then $Q_j^{\le L}$ contains no two consecutive integers and
\[
  \bigl|\{1\le m\le M_j:\ell_j^\to(m)>L\}\bigr|
  \ge\left\lfloor\frac{M_j+1}{2}\right\rfloor.
\]
\end{theorem}

\begin{proof}
Subtract representations of $u h_j$ and $v h_j$.  The resulting integer coefficient vector has $\ell_1$ norm at most $2L$, so subadditivity of the torus norm proves the first display.  Under the strict phase inequality, two short-representable multipliers cannot differ by one.  An independent set in the path on $\{0,1,\ldots,M_j\}$ has size at most
$\lceil(M_j+1)/2\rceil$.  Since $0\in Q_j^{\le L}$, every long multiplier
is positive, and hence
\[
  \bigl|\{1\le m\le M_j:\ell_j^\to(m)>L\}\bigr|
  =(M_j+1)-|Q_j^{\le L}|
  \ge \left\lfloor\frac{M_j+1}{2}\right\rfloor.
\]
\end{proof}

A Fourier character with small error on the earlier gaps but large error on
the current gap therefore leaves at least half of the multipliers
$1,\ldots,M_j$ either unrepresented or expensive to represent.  The converse
argument uses the number of represented multipliers rather than the length of
each representation.

\subsection{Few missing multipliers give a short integer relation}

If at most $g$ positions are missing from an interval of $2g+2$ consecutive
integers, two represented positions must be adjacent.  Subtracting their
nonnegative representations yields an integer representation of $h_j$.
\Cref{lem:signed-transfer} then bounds the Fourier error, without any
assumption that $Q_j$ is cofinite.

\begin{theorem}[Few missing multipliers give a short integer representation]\label{thm:ray-genus-transfer}
Fix $2\le j\le p$ and suppose the finite range is sufficiently long:
\[
  M_j=\left\lfloor\frac{n-1}{h_j}\right\rfloor
  \ge 2g_j^{(n)}+1.
\]
Then there exist consecutive multipliers $q,q+1\in Q_j$ with
$1\le q\le 2g_j^{(n)}$, and the minimum coefficient length satisfies
\[
  \boxed{\lambda_j(H)\le 4g_j^{(n)}-1.}
\]
By \cref{lem:signed-transfer}, for every $1\le k<n$,
\[
  \normfrac{kh_j/n}\le(4g_j^{(n)}-1)\,\delta_{j-1}(k).
\]
\end{theorem}

\begin{proof}
The set $Q_j\cap[0,M_j]$ contains $0$ and has $g_j^{(n)}$ missing
elements in $[1,M_j]$.  Consider the path graph on
$\{0,1,\ldots,2g_j^{(n)}+1\}$: it has $2g_j^{(n)}+2$ vertices, of which
at most $g_j^{(n)}$ are gaps.  A maximum independent set in a path on $N$
vertices has size $\lceil N/2\rceil$, so any set of more than
$\lceil(2g_j^{(n)}+2)/2\rceil=g_j^{(n)}+1$ vertices must contain two
consecutive ones.  Since $Q_j\cap[0,2g_j^{(n)}+1]$ has at least
$2g_j^{(n)}+2-g_j^{(n)}=g_j^{(n)}+2$ elements (including $0$), there exist
consecutive $q,q+1\in Q_j$.  Strict decrease of the gaps gives
$1\notin Q_j$: no nonnegative combination of earlier gaps, all larger than
$h_j$, can equal $h_j$.  The window condition implies $M_j\ge1$, so this
also shows $g_j^{(n)}\ge1$ and excludes the pair $(0,1)$.  Thus $q\ge1$.

Since $qh_j,\,(q+1)h_j\in S_{j-1}$ and every generator of $S_{j-1}$ is
strictly greater than $h_j$, a positive representation of $qh_j$ has
length $L_q$ satisfying $L_q(h_j+1)\le L_q\min_{i<j}h_i\le qh_j$, so
$L_q\le q-1$ (and similarly $L_{q+1}\le q$).  Subtracting the two
representations gives an integer representation of $h_j$ of length at most
$L_q+L_{q+1}\le(q-1)+q=2q-1\le 4g_j^{(n)}-1$.

The phase inequality follows from \cref{lem:signed-transfer}.
\end{proof}

The estimate in \cref{prop:truncated-genus-refinement} converts the number of
missing multipliers into a bound involving their weighted count.

\begin{corollary}[Bound from the weighted count]\label{cor:defect-transfer}
Under the same window condition $M_j\ge2g_j^{(n)}+1$,
\[
  \lambda_j(H)\le 1+2\sqrt{1+8D_j(n)/h_j}.
\]
\end{corollary}

\begin{proof}
By \cref{prop:truncated-genus-refinement},
$g_j^{(n)}\le(1+\sqrt{1+8D_j(n)/h_j})/2$.  Substitute into the bound
$\lambda_j\le4g_j^{(n)}-1$ and simplify.
\end{proof}

\paragraph{When the finite range is shorter}
If $M_j<2g_j^{(n)}+1$, the upper bound still uses
$g_j^{(n)}\le M_j$, or it bounds all passes after a chosen prefix through
$U_s(n)$.  The displayed condition is needed only to deduce the bound on
$\lambda_j$ from the number of missing multipliers.

\subsection{When the multiplier set is a numerical semigroup}

If $Q_j$ is a numerical semigroup, its conductor supplies two consecutive
represented multipliers and removes the finite-range condition.  This occurs
exactly when
$\gcd(h_1,\ldots,h_{j-1})\mid h_j$.

\begin{corollary}[The cofinite case]\label{cor:full-ray-transfer}
If $Q_j$ is a numerical semigroup with conductor $c$ and genus $g=g(Q_j)$,
then $c,c+1\in Q_j$ and
\[
  \boxed{\lambda_j(H)\le 2c-1\le 4g-1.}
\]
The same phase lemma then gives, for every $1\le k<n$,
\[
  \normfrac{kh_j/n}\le(2c-1)\,\delta_{j-1}(k).
\]
\end{corollary}

\begin{proof}
By definition of the conductor, every integer $\ge c$ lies in $Q_j$;
in particular $c,c+1\in Q_j$.  Since $c\ge1$ and all generators of
$S_{j-1}$ exceed $h_j$, the same length estimate as in
\cref{thm:ray-genus-transfer} gives $\lambda_j\le2c-1$.  To see the
conductor--genus bound directly, put $F=c-1$.  The map
$s\mapsto F-s$ sends every nongap $s\in Q_j\cap[0,F]$ injectively to a
gap: otherwise both $s$ and $F-s$ would lie in $Q_j$, forcing
$F\in Q_j$.  There are $c-g$ such nongaps and $g$ gaps, so
$c-g\le g$, hence $c\le2g$.
\end{proof}

A small generalized Frobenius count $g(Q_j)$ therefore bounds the Fourier
error on the current gap in terms of the errors on earlier gaps.

\subsection{Why the converse fails}

An integer relation records cancellation.  A nonnegative representation also
depends on the magnitudes and residue classes of the available gaps.  The
integer relation loses this information, so the reverse implication is false.

\begin{proposition}[A short integer representation does not bound missing multipliers]\label{prop:length-genus-nogo}
There is no function $F:\mathbb N\to\mathbb N$ such that, for all
$n,H,j$, and $L\in\mathbb N$,
\[
  \lambda_j(H)\le L
  \quad\Longrightarrow\quad
  g_j^{(n)}\le F(L).
\]
In fact, this fails even when $\lambda_j(H)=2$, $Q_j$ is a numerical
semigroup, and every represented multiplier in the finite range has
nonnegative representation length one.
\end{proposition}

\begin{proof}
For $q\ge3$, take $n=2q$ and the gap sequence $H=(q+1,q,1)$.  At entry to
the final pass, $Q_3=\langle q,q+1\rangle$ with $M_3=2q-1$.  The only
positive elements of $Q_3$ in $[1,M_3]$ are $q$ and $q+1$, each of length
one, while $g_3^{(n)}=2q-3$.  Moreover,
$1=(q+1)-q$, so $\lambda_3(H)\le2$; the reverse inequality follows from
\cref{lem:signed-transfer}.  Hence $\lambda_3(H)=2$ while the number of
missing multipliers tends to infinity with $q$, ruling out $F$ already at
$L=2$.
\end{proof}

The implication proved above runs from few missing multipliers to a short
integer representation and then to a bound on Fourier error.  Information
about magnitudes and residue classes supplies the missing converse in two
useful settings.  Two coprime represented multipliers give a quadratic bound
(\cref{prop:short-coprime-bridge}); a complete residue system represented with
bounded length and scale gives a linear bound
(\cref{thm:apery-bridge}).  Finding the weakest assumptions under which a
short integer representation bounds the number of missing multipliers remains
open.

\subsection{Using two generators or an Ap\'ery set}

A coprime pair already gives a quadratic genus bound.

\begin{proposition}[Two coprime represented multipliers]\label{prop:short-coprime-bridge}
Suppose $a,b\in Q_j$ are coprime,
$\max\{a,b\}\le B$, and
$\ell_j^\to(a),\ell_j^\to(b)\le L$.  Then
\[
  g_j^{(n)}\le\frac{(a-1)(b-1)}2<\frac{B^2}{2}
\]
and, for every $1\le k<n$,
$\normfrac{k h_j/n}<2BL\,\delta_{j-1}(k)$.
\end{proposition}

\begin{proof}
The genus bound follows from
$\langle a,b\rangle\subseteq Q_j$ and Sylvester's theorem.  A B\'ezout
identity $ua+vb=1$ with $|u|<b$ and $|v|<a$, together with nonnegative
representations of $ah_j$ and $bh_j$, gives an integer representation of
$h_j$ of length less than $L(a+b)\le2BL$.
\end{proof}

For a linear rather than quadratic bound on missing multipliers, one needs
complete coverage of the residue classes.

\begin{theorem}[Upper bound from Ap\'ery representatives]\label{thm:apery-bridge}
Fix $2\le j\le p$ and $L,B\ge1$.  Suppose there is an integer
$a\in Q_j$ with $a\ge2$ such that $a h_j$ has a nonnegative representation
using at most $L$ earlier gaps, each no larger than $B h_j$.  Suppose also
that for every residue
$r\in\{0,\ldots,a-1\}$ there is a multiplier $w_r\in Q_j$ such that
\[
  w_r\equiv r\pmod a
\]
and $w_r h_j$ has a nonnegative representation using at most $L$ earlier
gaps, each no larger than $B h_j$; take $w_0=0$.  Then $Q_j$ is a numerical
semigroup and
\[
  \operatorname{cond}(Q_j)\le LB,
  \qquad
  g_j^{(n)}\le g(Q_j)<LB.
\]
Moreover,
\[
  \lambda_j(H)\le \min\left\{2LB-1,\; L\left(1+\frac{LB}{a}\right)\right\},
\]
and hence, for every $1\le k<n$,
\[
  \normfrac{k h_j/n}
  \le \min\left\{2LB-1,\; L\left(1+\frac{LB}{a}\right)\right\}
       \delta_{j-1}(k).
\]
\end{theorem}

\begin{proof}
Every representation in the hypotheses satisfies
\[
  w_r h_j=\sum_{i<j}x_i h_i
  \le \|x\|_1 B h_j\le LBh_j,
\]
so $w_r\le LB$.  Because $a\in Q_j$ and every residue modulo $a$ occurs in $Q_j$, the additive monoid $Q_j$ is cofinite and hence is a numerical semigroup.  Equivalently, it has gcd one: a larger gcd would leave infinitely many integers outside $Q_j$.  Let $v_r$ be the least element of $Q_j$ congruent to $r$ modulo $a$.  Then $v_r\le w_r$, and
$\{v_0,\ldots,v_{a-1}\}$ is the Ap\'ery set of $Q_j$ with respect to $a$.  Counting the missing integers separately in each residue class gives the standard identity \cite{rosalesgarcia2009}
\[
  g(Q_j)
  =\frac1a\sum_{r=0}^{a-1}v_r-\frac{a-1}{2}
  \le\frac1a\sum_{r=0}^{a-1}w_r-\frac{a-1}{2}
  <LB.
\]
Since $G_j(n)\subseteq\mathbb N\setminus Q_j$, one has
$g_j^{(n)}\le g(Q_j)$.
If $W=\max_r v_r$, then every integer greater than $W-a$ lies in $Q_j$, so
$\operatorname{cond}(Q_j)\le W-a+1\le LB$.

Finally, $w_1\ge1$, so write $w_1=ta+1$ with $t\in\mathbb N_0$; the case
$t=0$ is allowed.  Since $w_1\le LB$, one has
$t\le(LB-1)/a<LB/a$.  Subtract $t$ copies of a length-$L$ representation of $a h_j$ from a length-$L$ representation of $w_1h_j$.  This is an integer representation of $h_j$ of length at most
  $L(1+t)\le L(1+LB/a)$.  Alternatively,
  \cref{cor:full-ray-transfer} and
  $\operatorname{cond}(Q_j)\le LB$ give
  $\lambda_j(H)\le2\operatorname{cond}(Q_j)-1\le2LB-1$.
  Taking the minimum proves the bound on $\lambda_j$, and the Fourier-error
  bound follows from \cref{lem:signed-transfer}.
\end{proof}

\paragraph{Combining the Ap\'ery-set bounds across passes}
\Cref{thm:apery-bridge} bounds both the number of missing multipliers and
$\lambda_j$.  For a pass satisfying its hypotheses, use $L_jB_j$; for any
other pass before a chosen cutoff $s$, use the elementary bound
$g_j^{(n)}\le M_j$; and for all later passes, use $U_s(n)/n$.  The smallest
total obtained in this way is
\[
\boxed{
  \mathfrak A_n(H)
  =
  \min_{\substack{0\le s\le p,\ J\subseteq[1,s]\\
                  \text{Ap\'ery data available on }J}}
  \left[
    \sum_{j\in J}L_jB_j
    +\sum_{\substack{1\le j\le s\\j\notin J}}M_j
    +\frac{U_s(n)}{n}
  \right].}
\]
The minimum is always finite.

\begin{corollary}[Aggregated Ap\'ery-set upper bounds]
\label{cor:apery-runtime}
Let $J$ be a set of passes satisfying the hypotheses of
\cref{thm:apery-bridge}, with parameters $(L_j,B_j)$, and use
$g_j^{(n)}\le M_j$ on the remaining passes.  Then
\[
\begin{aligned}
  C_n(H)
  &\le n\left(
      p+\sum_{j\in J}L_jB_j
        +\sum_{j\notin J}M_j
    \right) \\
  &\qquad
    -\frac12\sum_{j=1}^{p}
      h_j g_j^{(n)}(g_j^{(n)}+1) \\
  &\le n\left(
      p+\sum_{j\in J}L_jB_j
        +\sum_{j\notin J}M_j
    \right).
\end{aligned}
\]
In particular, a polylogarithmic value of the displayed total is sufficient
for a polylogarithmic factor above linear cost.  More generally,
\[
  W_n(H)\le n\mathfrak A_n(H),
  \qquad
  C_n(H)\le n\bigl(p+\mathfrak A_n(H)\bigr).
\]
\end{corollary}

\begin{proof}
The first display follows by applying
\cref{prop:truncated-genus-refinement} on every pass, using
$g_j^{(n)}<L_jB_j$ on $J$ and $g_j^{(n)}\le M_j$ elsewhere, and then adding
the comparison overhead from \cref{thm:directed-upper}.  Dropping the
nonnegative correction gives the second inequality.  For the general form,
fix a cutoff $s$, a set $J\subseteq[1,s]$, and data occurring in the minimum
defining $\mathfrak A_n(H)$.  The same argument bounds the exchanges in the
first $s$ passes by
\[
 n\left(
   \sum_{j\in J}L_jB_j+
   \sum_{\substack{1\le j\le s\\j\notin J}}M_j
 \right).
\]
The cutoff clause of \cref{thm:directed-upper} bounds all later exchanges by
$U_s(n)$.  Add the compulsory comparison overhead $np$ and minimize.
\end{proof}

\subsection{Exact evaluation from an Ap\'ery set}

The same finite data can also evaluate the weighted count $D_j(n)$ exactly.
This is useful when applying the upper bound, but it is not needed for the
main product-grid theorem.  The formulas below are standard Ap\'ery-set and
Sylvester-sum identities \cite{rosalesgarcia2009,brownshiue1993}.

\begin{proposition}[Exact weighted count from an Ap\'ery set]
\label{prop:apery-exact}
Suppose that $Q_j$ is a numerical semigroup and every positive integer outside
$Q_j$ lies in $[1,M_j]$.  Fix $a\in Q_j$, $a\ge2$, and write
\[
  \operatorname{Ap}(Q_j,a)=\{v_0=0,v_1,\ldots,v_{a-1}\},
  \qquad v_r\equiv r\pmod a.
\]
Then
\[
  g(Q_j)=\frac1a\sum_{r=0}^{a-1}v_r-\frac{a-1}{2},
\]
\[
  n_1(Q_j):=\sum_{m\in\mathbb N_{\ge1}\setminus Q_j}m
  =\frac1{2a}\sum_{r=0}^{a-1}v_r^2
   -\frac12\sum_{r=0}^{a-1}v_r
   +\frac{a^2-1}{12},
\]
and
\[
  \boxed{D_j(n)=n\,g(Q_j)-h_j\,n_1(Q_j).}
\]
Thus the $a$ Ap\'ery values determine $D_j(n)$ in $O(a)$ arithmetic
operations.
\end{proposition}

\begin{proof}
Write $v_r=q_ra+r$.  The missing positive integers in residue class $r$ are
$r,r+a,\ldots,r+(q_r-1)a$.  Summing their number and their values over all
residue classes gives the two displayed formulas.  The assumption on
$M_j$ means that $G_j(n)$ contains every positive integer outside $Q_j$;
substitution in the definition of $D_j(n)$ proves the last identity.
\end{proof}

\FloatBarrier

\section{Examples and limits of the general bounds}
\label{app:additional-consequences}

The next two families test the limits of the preceding bounds.  The first
shows that a global quantity based on all earlier gaps can greatly overestimate
the cost of one pass.  The second shows that the bound relating missing
multiples to the length of an integer relation has the best possible order.

\subsection{A global invariant can overestimate one pass}

The next family shows that ordinary global semigroup invariants (genus and
conductor) can remain large while the set of represented multiples of the
current gap omits only two integers.

\begin{proposition}[A global-to-pass separation]\label{prop:affine-ray-separation}
For $t\ge3$, set $n=16t^2$ and
\[
  H_t=(6t+1,5t,4t,3t,1).
\]
Let $S_t=\langle4t,5t,6t+1\rangle$ be the prefix semigroup before the $3t$-pass.  Then
\[
  g(S_t)=3t^2+3t=\Theta(n),
  \qquad
  \operatorname{cond}(S_t)=6t^2+6t=\Theta(n),
\]
yet
\[
  Q_t:=S_t/(3t)=\mathbb N_0\setminus\{1,2\},
  \qquad
  g(Q_t)=2,
  \qquad
  \operatorname{cond}(Q_t)=3.
\]
Moreover, the weighted count and minimum integer-representation length satisfy
\[
  D_4(n)=32t^2-9t=\Theta(n),
  \qquad
  \lambda_4(H_t)=3.
\]
\end{proposition}

\begin{proof}
Let $A=\langle4,5\rangle$; its positive gaps are
$\{1,2,3,6,7,11\}$.  Write $x=qt+r$ with $0\le r<t$.  We first show
\[
  x=qt+r\in S_t
  \quad\Longleftrightarrow\quad
  q-6r\in A.
\]
Indeed, in a representation
$x=4ta+5tb+(6t+1)c$, reduction modulo $t$ gives $c=r+kt$.  Hence
$q-6r=4a+5b+k(6t+1)\in A$; here $6t+1\in A$ because $t\ge3$ and
$\langle4,5\rangle$ contains every integer at least $12$.
Conversely, $q-6r=4a+5b$ gives
$x=4ta+5tb+r(6t+1)$.

For a fixed $r$, the missing values have either $0\le q<6r$ or
$q=6r+e$ with $e\in\{1,2,3,6,7,11\}$.  Therefore
\[
  g(S_t)=\sum_{r=0}^{t-1}(6r+6)=3t^2+3t.
\]
The largest gap occurs at $r=t-1$ and $e=11$ and equals
$6t^2+6t-1$, proving the conductor formula.  Taking $r=0$ and
$x=3tm$ shows that $m\in Q_t$ exactly when $3m\in A$; the only missing
positive multiples of $3$ in $A$ are $3$ and $6$.  Thus
$Q_t=\mathbb N_0\setminus\{1,2\}$.

The identity $3t=2(4t)-5t$ gives $\lambda_4(H_t)\le3$.  An integer
representation of length at most two would make $3t$ either one earlier gap,
the sum of two earlier gaps, twice one earlier gap, or the absolute difference
of two earlier gaps.  The positive sums are at least $8t$, while the positive
differences among $4t,5t,6t+1$ are $t,t+1,2t+1$; none equals $3t$ for
$t\ge3$.  Hence $\lambda_4(H_t)=3$.  In particular, no earlier gap is an
integer multiple of $3t$, although a short relation with cancellation exists.

For $D_4(n)$, the finite multiplier range $[1,M_4]$ has
$M_4=\lfloor(16t^2-1)/(3t)\rfloor$ elements.  Since $Q_t$ misses only
$\{1,2\}$, which both lie in this window,
$g_4^{(n)}=2$ and
\[
  D_4(n)=\sum_{m\in\{1,2\}}(n-m\cdot3t)
  =2n-3t(1+2)=32t^2-9t.
\]
\end{proof}

\begin{remark}
The generalized Frobenius invariant $n_{3t}(4t,5t,6t+1)$ also equals two,
coinciding with $g(Q_t)$.  The separation here is between the ordinary global
semigroup genus $g(S_t)=\Theta(n)$ and the pass-specific genus
$g(Q_t)=O(1)$.
\end{remark}

\subsection{Why the linear relation bound is best possible}

The linear bound $\lambda_j=O(g_j)$ in \cref{thm:ray-genus-transfer} cannot
be improved to $o(g_j)$ in general.

\begin{proposition}[The linear bound is tight]\label{prop:genus-transfer-sharpness}
For every $h\ge3$ with $h\not\equiv2\pmod5$, set
\[
  n_h=6h^2,
  \qquad
  H_h=(3h-1,2h+1,h,1).
\]
At the $h$-pass, the multiplier set is
\[
  Q_h=\langle2h+1,3h-1\rangle/h
      =\langle5,2h+1,3h-1\rangle,
\]
and
\[
  g_3^{(n_h)}=g(Q_h)=3h-2,
  \qquad
  \lambda_3(H_h)=
  \begin{cases}
    h,&h\equiv0,4\pmod5,\\
    2h,&h\equiv1,3\pmod5.
  \end{cases}
\]
Moreover $M_3=6h-1\ge2g_3^{(n_h)}+1$.  Hence
$\lambda_3(H_h)=\Theta(g_3^{(n_h)})$ under the finite-range condition of
\cref{thm:ray-genus-transfer}.
\end{proposition}

\begin{proof}
Put $A=2h+1$ and $B=3h-1$.  The assumption on $h$ is equivalent to
$\gcd(A,B)=\gcd(h-2,5)=1$.  If $mh=xA+yB$, then
\[
  mh=h(2x+3y)+(x-y),
\]
so $h$ divides $x-y$.  According to its sign, this gives either
$m=5y+kA$ or $m=5x+kB$.  Conversely, the identities
$5h=A+B$, $Ah=hA$, and $Bh=hB$ show that each generator on the right is
in the quotient.  Hence
$Q_h=\langle5,A,B\rangle$.

The numbers $A,2A,B,2B$ occupy the four nonzero residue classes modulo $5$.
They are the nonzero elements of $\operatorname{Ap}(Q_h;5)$: a combination
of at least three copies of $A$ and $B$ is larger than $2B$, while the only
combinations using at most two copies are
$0,A,B,2A,A+B,2B$, and $A+B=5h$ has residue zero.  The standard Ap\'ery
formulas therefore give
\[
  g(Q_h)=\frac{3(A+B)}5-2=3h-2,
  \qquad \operatorname{cond}(Q_h)=2B-4=6h-6.
\]
Since $M_3=6h-1$, the finite range contains every gap of $Q_h$.

It remains to minimize $|u|+|v|$ subject to $uA+vB=h$.  All solutions have
\[
  u=\frac{1-Bk}{5},\qquad v=\frac{1+Ak}{5},
  \qquad Bk\equiv1\pmod5.
\]
For nonzero $k$, the two coefficients have opposite signs and
$|u|+|v|=h|k|$.  The least possible $|k|$ is $1$ when
$h\equiv0,4\pmod5$ and $2$ when $h\equiv1,3\pmod5$, which proves the
claimed formula for $\lambda_3(H_h)$.
\end{proof}

Thus \cref{thm:ray-genus-transfer} is tight up to constant factors.

\FloatBarrier

\section{Computational checks}\label{app:verification}

The companion artifact implements the finite exhaustive checks and seeded
family calculations summarized below.  Its distribution directory is
\[
  \texttt{chronological-shellsort-artifact/}.
\]
From the artifact root, run
\texttt{python} \path{run_verification.py}.
The $17$ programs exhaustively check persistence and the upper bound for
$n\le8$.  They also test seeded families for integer-representation length,
the Ap\'ery-set bound, the combined Fourier estimate, weighted counts, the
global-to-pass separation, and tightness.  All checks
pass.  The artifact's \path{README.md} records the deterministic generation
rules, verification scope, and mapping from programs to stable manuscript
labels.

\bibliographystyle{elsarticle-num}
\bibliography{refs}

\begin{thebibliography}{1}
\expandafter\ifx\csname url\endcsname\relax
  \def\url#1{\texttt{#1}}\fi
\expandafter\ifx\csname urlprefix\endcsname\relax\def\urlprefix{URL }\fi
\expandafter\ifx\csname href\endcsname\relax
  \def\href#1#2{#2} \def\path#1{#1}\fi

\bibitem{pratt1971}
V.~R. Pratt, Shellsort and sorting networks, Ph.D. thesis, Stanford University,
  reprinted by Garland, 1979 (1971).

\bibitem{poonen1993}
B.~Poonen, The worst case in {Shellsort} and related algorithms, Journal of
  Algorithms 15~(1) (1993) 101--124.
\newblock \href {https://doi.org/10.1006/jagm.1993.1032}
  {\path{doi:10.1006/jagm.1993.1032}}.

\bibitem{plaxtonsuel1997}
C.~G. Plaxton, T.~Suel, Lower bounds for {Shellsort}, Journal of Algorithms
  23~(2) (1997) 221--240.
\newblock \href {https://doi.org/10.1006/jagm.1996.0825}
  {\path{doi:10.1006/jagm.1996.0825}}.

\bibitem{shell1959}
D.~L. Shell, A high-speed sorting procedure, Communications of the ACM 2~(7)
  (1959) 30--32.
\newblock \href {https://doi.org/10.1145/368370.368387}
  {\path{doi:10.1145/368370.368387}}.

\bibitem{sedgewick1986}
R.~Sedgewick, A new upper bound for {Shellsort}, Journal of Algorithms 7~(2)
  (1986) 159--173.
\newblock \href {https://doi.org/10.1016/0196-6774(86)90001-5}
  {\path{doi:10.1016/0196-6774(86)90001-5}}.

\bibitem{incerpi1985}
J.~Incerpi, R.~Sedgewick, Improved upper bounds on {Shellsort}, Journal of
  Computer and System Sciences 31~(2) (1985) 210--224.
\newblock \href {https://doi.org/10.1016/0022-0000(85)90042-X}
  {\path{doi:10.1016/0022-0000(85)90042-X}}.

\bibitem{rosalesgarcia2009}
J.~C. Rosales, P.~A. Garc{\'i}a-S{\'a}nchez, Numerical Semigroups, Vol.~20 of
  Developments in Mathematics, Springer, 2009.
\newblock \href {https://doi.org/10.1007/978-1-4419-0160-6}
  {\path{doi:10.1007/978-1-4419-0160-6}}.

\bibitem{zang2026}
Z.~Zang, \href{https://arxiv.org/abs/2607.08997}{Improved lower bounds of the
  time complexity of {Shellsort}}, arXiv:2607.08997v1, 10 July 2026 (2026).
\newblock \href {http://arxiv.org/abs/2607.08997} {\path{arXiv:2607.08997}},
  \href {https://doi.org/10.48550/arXiv.2607.08997}
  {\path{doi:10.48550/arXiv.2607.08997}}.
\newline\urlprefix\url{https://arxiv.org/abs/2607.08997}

\bibitem{brownshiue1993}
T.~C. Brown, P.~J.-S. Shiue, A remark related to the {Frobenius} problem,
  Fibonacci Quarterly 31~(1) (1993) 32--36.
\newblock \href {https://doi.org/10.1080/00150517.1993.12429318}
  {\path{doi:10.1080/00150517.1993.12429318}}.

\end{thebibliography}

\end{document}